\documentclass[11pt,twoside]{article}
\usepackage{amsmath}
\usepackage{graphicx}
\usepackage{amsthm}
\usepackage{amssymb}

\newtheorem{thm}{Theorem}[section]

\newtheorem{defn}[thm]{Definition}
\newtheorem{rem}[thm]{Remark}

\title{\bf HYPERGRAPHS IN THE CHARACTERIZATION OF REGULAR-VINE COPULA STRUCTURES}

\author{\\ Edith KOV\'{A}CS\\ 
Budapest College of Management \\ and Tam\'{a}s SZ\'{A}NTAI\\ Budapest University of Technology and Economics}

\date{}

\begin{document}

\thispagestyle{plain}

\maketitle

\vspace{-10cm}

\begin{quotation}
\begin{center}
 Proceedings of the $XIII^{th}$ Conference on Mathematics and its Applications\\
University "Politehnica" of Timisoara\\
November, 1-3, 2012
\end{center}
\end{quotation}
\vspace{7.5cm}

\begin{abstract}
Vine copulas constitute a flexible way for modeling of dependences using only pair copulas as building blocks. The pair-copula constructions introduced by Joe (1997) are able to encode more types of dependences in the same time since they can be expressed as a product of different types of bi-variate copulas. The Regular-vine structures (R-vines), as pair copulas corresponding to a sequence of trees, have been introduced by Bedford and Cooke (2001, 2002) and further explored by Kurowicka and Cooke (2006). The complexity of these models strongly increases in larger dimensions.  Therefore the so called truncated R-vines were introduced in Brechmann et al. (2012). 
In this paper we express the Regular-vines using a special type of hypergraphs, which encodes the conditional independences. 
\end{abstract}
\bigskip

\section{Introduction}
\label{sec:sec1}
\smallskip

Copulas in general are known to be useful tool for modeling
multivariate probability distributions since they serve as a link between univariate marginals.
In this paper we show how conditional independences can be utilized in the expression of multivariate copulas.
Regarding to this we prove a theorem which links to a multivariate probability distribution assigned to a junction tree the so called junction tree copula. A hard practical problem is finding those lower dimensional copulas which are involved in the expression of the junction tree copula.

The paper Aas et al. (2009) calls the attention on the fact that ''conditional independence may reduce the number of the
pair-copula decompositions and hence simplify the construction''. 
In this paper the importance of choosing a good
factorisation which takes advantage from the conditional independence
relations between the random variables is pointed out. 
In the present paper we give a method for
findig that pair copula construction which exploits the conditional
independences between the variables.

\section{Preliminaries}
\label{sec:sec2}
\smallskip

In this section we introduce some concepts used in graph theory and
probability theory that we need throughout the paper and present how these
can be linked to each other. For a good overview see Lauritzen and Spiegelhalter (1988).

We first present the acyclic hypergraphs and junction trees. 
Then we introduce the cherry trees as a special type of junction trees.
We finish this part with the multivariate joint
probability distribution associated to junction trees.

Let $V=\left\{ 1,\ldots ,d\right\} $ be a set of vertices and $\Gamma $ a
set of subsets of $V$ called {\it set of hyperedges}. A \textit{hypergraph}
consists of a set $V$ of vertices and a set $\Gamma $ of hyperedges. We
denote a hyperedge by $K_i$, where $K_i$ is a subset of $V$. If two vertices
are in the same hyperedge they are connected, which means, the hyperedge of a
hyperhraph is a complete graph on the set of vertices contained in it.

The \textit{acyclic}{\normalsize \ }\textit{hypergraph} is a special type of
hypergraph which fulfills the following requirements:

\begin{itemize}
\item  Neither of the hyperedges of $\Gamma $ is a subset of another hyperedge.

\item  There exists a numbering of edges for which the \textit{running
intersection property} is fullfiled: $\forall j\geq 2\quad \ \exists \ i<j:\
K_i\supset K_j\cap \left( K_1\cup \ldots \cup K_{j-1}\right) $. (Other
formulation is that for all hyperedges $K_i$ and $K_j$ with $i<j-1$,
$K_i \cap K_j \subset K_s \ \mbox{for all} \ s, i<s<j$.)
\end{itemize}

Let $S_j=K_j\cap \left( K_1\cup \ldots \cup K_{j-1}\right) $, for $j>1$ and $%
S_1=\phi $. Let $R_j=K_j\backslash S_j$. We say that $S_j$\textit{separates} 
$R_j$ from $\left( K_1\cup \ldots \cup K_{j-1}\right) \backslash S_j$, and
call $S_j$ separator set or shortly separator.

Now we link these concepts to the terminology of junction trees.

The {\it junction tree} is a special tree stucture which is equivalent to the
connected acyclic hypergraphs (Lauritzen and Spiegelhalter (1988)). The nodes of the tree correspond
to the hyperedges of the connected acyclic hypergraph and are called clusters, the edges of the tree
correspond to the separator sets and called separators. The set of all
clusters is denoted by $\Gamma$, the set of all separators is denoted by %
$\mathcal{S}$. A junction tree $(V,\Gamma, S)$ is defined by the set of vertices $V$, the set of nodes $\Gamma$ called also set of clusters, and the set of separators $S$. The junction tree with the largest cluster containing $k$
variables is called \textit{k-width junction tree}. 

The concept of \textit{junction tree probability distribution} is related to
the junction tree graph and to the global Markov property of the graph. A
junction tree probability distribution is defined as a product and fraction of
marginal probability distributions as follows:

\begin{equation}\label{eq:eq1}
P\left( \mathbf{X}\right) =\dfrac{\prod\limits_{C\in \Gamma}P\left(
\mathbf{X}_C\right) }{\prod\limits_{S\in \mathcal{S}}\left[ P\left( \mathbf{X}_S\right)
\right] ^{\nu _S-1}},
\end{equation}
where $\Gamma$ is the set of clusters of the
junction tree, $\mathcal S$ is the set of separators, $\nu_S$ is the number of those clusters
which contain the separator $S$. We emphasize here that the equalities written as
$P(\mathbf{X})=f(P(\mathbf{X}_K), K\in \Gamma)$, where $f: \Omega_{\mathbf{X}}\rightarrow R$
hold for any possible realization of $\mathbf{X}$.

In Buksz\'{a}r and Pr\'{e}kopa (2001) and Buksz\'{a}r and Sz\'{a}ntai (2002) there were introduced the so called $t$-cherry tree graph structures. Since these can be regarded as a special type of junction tree we can give now the following definition. In this paper we will call this structure simply cherry tree as this does not cause any confusion.

\begin{defn}
\label{def:def2.2} 
We call $k$-th order cherry tree the junction tree with all clusters of size $k$ and all separators of size $k-1$.
\end{defn}

We will denote by $\mathcal{C}_{\mbox{ch}}$ and $\mathcal{S}_{\mbox{ch}}$ subsets of $V$, the
set of clusters and separators of the cherry junction tree.

\begin{defn}
\label{def:def2.3}
(Sz\'{a}ntai, Kov\'{a}cs (2012))
The probability distribution assigned to a cherry tree is called {\it cherry tree probability distribution}.
\end{defn}
The marginal density functions involved in 
Formula (\ref{eq:eq1}) are marginal probability distributions of $P\left( \mathbf{X}\right) $.

We summarize here some of our results in Sz\'antai and Kov\'acs (2008) which will be used later, in Section 5.
Let $P(\mathbf{X})$  be a joint probability distribution which is approximated by a $k$-width junction tree pd.

\begin{thm}\label{theo:theo2.4}
(Theorem 4 in Sz\'antai and Kov\'acs (2008)): A $k$-width junction tree pd can be transformed into a $k$-th order cherry tree pd which gives at least as good approximating pd as the $k$-width junction tree did.
\end{thm}

In Sz\'antai and Kov\'acs (2008) there is given a constructive algorithm, called Algorithm 2 which performs the transformation claimed in the Theorem \ref{theo:theo2.4}.

\begin{thm}\label{theo:theo2.6}
(Theorem 7 in Sz\'antai and Kov\'acs (2008)): The $(k+1)$-th order cherry tree pd obtained by the constructive algorithm starting from the best approximating $k$-th order cherry tree pd, approximates at least as good $P(\mathbf{X})$ as the $k$-th order did.
\end{thm}

\section{The multivariate copula associated to a junction tree probability distribution. The cherry-tree copulas.}
\label{sec:sec3}
\smallskip

In the following we introduce the so called cherry-tree copula, which incorporates some of the conditional independences between the variables.

We will use the following notations:

\begin{tabular}{lcl}
$F_{i,j|D}$ & -- & the conditional probability distribution function of $X_i$ and $%
X_j$ given $\mathbf{X}_D$;\\
$f_{i,j|D}$ & -- & the conditional probability density function of $X_i$ and $X_j$
given $\mathbf{X}_D$, \\
$c_{i,j|D}$ & -- & the conditional copula density corresponding to  $f_{i,j|D}$,
\end{tabular}
where $D\subset V;i,j\in V\backslash D$.

In Aas et al. (2009) the inference of pair-copula decomposition is depicted in three parts: 
\vspace{-3mm}
\begin{itemize}
\item  The selection of a specific factorization (structure);\vspace{-3mm}
\item  The choice of pair-copula types;\vspace{-3mm}
\item  The estimation of parameters of the chosen pair-copulas.
\end{itemize}
\vspace{-3mm}

This paper deals with finding a good factorization which exploits some of the conditional independences between the random variables.

There are many papers dealing with selecting specific
Regular-vines as C-vine or D-vine see for example Aas et al. (2009).

In this section we give a theorem which assures the existence of a special type of copula density, which can be assigned to a junction tree graph structure.
Let us consider a random vector $\mathbf{X}=(X_1, X_2, \ldots, X_d)^{T}$, with the set of indices $V=\{1, 2, \ldots, d\}$.  Let $(V,\Gamma,S)$  be a junction tree defined on the vertex set $V$, by the cluster set $\Gamma$, and the separator set $S$.
\begin{thm}\label{theo:theo3.1} 
The copula density function associated to a junction tree
probability distribution
\begin{equation}\label{eq:eq5}
f_{\mathbf{X}}\left(  \mathbf{x}\right)
=\dfrac{\prod\limits_{K\in\Gamma}f_{\mathbf{X}_{K}}\left(  \mathbf{x}%
_{K}\right)  }{\prod\limits_{S\in\mathcal{S}}\left[  f_{\mathbf{X}_{S}}\left(
\mathbf{x}_{S}\right)  \right]  ^{v_{S}-1}},
\hspace{5mm} {\textit is \: given \: by} \hspace{5mm}
c_{\mathbf{X}}\left(  \mathbf{u}_{V}\right)  = \dfrac{\prod\limits_{K\in\Gamma
}c_{\mathbf{X}_{K}}\left(  \mathbf{u}_{K}\right)  }{\prod\limits_{S\in
\mathcal{S}}\left[  c_{\mathbf{X}_{S}}\left(  \mathbf{u}_{S}\right)  \right]
^{v_{S}-1}}.
\end{equation}
\end{thm}

\begin{proof}

\begin{equation}\label{eq:eq6}
f_{\mathbf{X}}\left(  \mathbf{x}\right)  =\dfrac{\prod\limits_{K\in\Gamma
}f_{\mathbf{X}_{K}}\left(  \mathbf{x}_{K}\right)  }{\prod\limits_{S\in
\mathcal{S}}\left[  f_{\mathbf{X}_{S}}\left(  \mathbf{x}_{S}\right)  \right]
^{v_{S}-1}}=\dfrac{\prod\limits_{K\in\Gamma}c_{\mathbf{X}_{K}}\left(
\mathbf{u}_{K}\right)  \cdot\prod\limits_{i_{k}\in K}^{{}%
}f_{X_{i_{k}}}\left(  x_{i_{k}}\right)  }{\prod\limits_{S\in\mathcal{S}%
}\left[  c_{\mathbf{X}_{S}}\left(  \mathbf{u}_{S}\right)
\cdot\prod\limits_{i_{k}\in S}^{{}}f_{X_{i_{k}}}\left(  x_{i_{k}}\right)
\right]  ^{v_{S}-1}}.
\end{equation}
\vspace{-3mm}

The question that we have to answer is how many times appears in the
nominator respectively in the denominator the probability density function
$f_{X_{i}}\left(  x_{i}\right)$ of each $X_{i}$ random variable.

Since $\bigcup\limits_{K\in\Gamma}\mathbf{X}_{K}=X$ for each random variable
$X_{i}$ $\epsilon X$, $f_{X_{i}}\left(  x_{i}\right)$ appears at least once
in the nominator.

Now we prove that in the junction tree over $X$ the number of clusters which
contain a variable $X_{i}$ is greater with 1 than the number of separators
which contain the same variable. This is true for all $i=1,\ldots,n$.
This means 
\vspace{-1mm}
\[
\#\left\{  K\in\Gamma|X_{i}\in X_{K}\right\}  = \# \left\{S\in\mathcal{S}|X_{i}\in X_{S}\right\}+1.
\]
For a variable $X_{i}$ we denote \#$\left\{ S\in\mathcal{S}|X_{i}\in
X_{S}\right\} $ by $t$.

Case: $t=0$

The statement is a consequence of the definition of junction tree, that is
the union of all clusters is $X$, so every variable have to appear at least
in one cluster. $X_{i}$ can not appear in two clusters, because in this case
there should exist a separator which contain $X_{i}$ too, and we supposed
that there is not such a separator $(t=0)$

Case: $t>0$

If two clusters contain the variable $X_{i}$, then every cluster from the
path between the two clusters contain $X_{i}$ (running intersection
property). From this results that the clusters containing $X_{i}$ are the
nodes of a connected graph, and this graph is a tree. If this tree contain $%
t $ separator sets then it contains $t+1$ clusters. All of these separators
contain $X_{i}$, and each separator connects two clusters. So there will be $%
t+1$ clusters that contain $X_{i}$.

Applying this result in formula (\ref{eq:eq6}) after simplification we obtain

\[
f_{\mathbf{X}}\left(  \mathbf{x}\right)  =\dfrac{\prod\limits_{K\in\Gamma
}c_{\mathbf{X}_{K}}\left(  \mathbf{u}_{K}\right)  \prod\limits_{i=1}%
^{d}f_{X_{i}}\left(  x_{i}\right)  }{\prod\limits_{S\in\mathcal{S}}\left[
c_{\mathbf{X}_{S}}\left(  \mathbf{u}_{S}\right)  \right]  ^{v_{S}-1}}.
\]

Dividing both sides by $\prod\limits_{i=1}^{d}f_{X_{i}}\left(  x_{i}\right)$ we obtain (\ref{eq:eq5}).
\end{proof}

\begin{defn}\label{def:def3.2}
The copula defined by (\ref{eq:eq5}) is called junction tree copula.
\end{defn}

We saw that if the conditional independence structure underlying the random
variables makes possible the construction of a junction tree, then the
multivariate copula density associated to the joint probability distribution
can be expressed as a product and fraction of lower dimensional copula
densities.

\begin{defn}\label{def:def3.3}
The copula density function associated to a cherry tree pd is called cherry-tree copula.
\end{defn}

\section{Regular-vine structure given by a sequence of cherry trees. Cherry-vine copula.}
\label{sec:sec4}

Pair-copula construction introduced by Joe (1997) is able to encode more
types of dependences in the same time since they can be expressed as a product of different types of bivariate copulas.
The Regular-vine structures were introduced by Bedford and Cooke (2001, 2002) and described in more details by Kurowicka and Cooke (2006). 

Now we give an alternative definition for Regular-vines by using the concept of cherry tree.

\begin{defn}\label{def:def4.1}
The cherry-vine structure is defined by a sequence of 
cherry junction trees $T_1,T_2,\ldots ,T_{d-1}$ as follows

\begin{itemize}
\item  $T_1$is a regular tree on $V=\left\{ 1,\ldots ,d\right\} $, the set
of edges is $E_1=\left\{ e_i^1=\left( l_i,m_i\right)\right. $ $\left.,i=1,\ldots ,d-1,\
l_i,m_i\in V\right\}$

\item  $T_2$ is the second order cherry junction tree on $%
V=\left\{ 1,\ldots ,d\right\} $, with the set of clusters $E_2=\left\{
e_i^2,i=1,\ldots ,d-1|e_i^2=e_i^1\right\} $ , $\left| e_i^1\right| =2$

\item  $T_k$ is one of the possible $k$-th order cherry junction tree on $%
V=\left\{ 1,\ldots ,d\right\}$, with the set of clusters $E_k=\left\{
e_i^k,i=1,\ldots ,d-k+1\right\}$ , where each $e_i^k,\left| e_i^k\right| =k$
is obtained from the union of two linked clusters in the $\left( k-1\right) $%
-th order cherry junction tree $T_{k-1}$.
\end{itemize}
\end{defn}

Next we define the pair copulas assigned to the cherry-vine structure given in Definition \ref{def:def4.1}

The copula densities $c_{l_i,m_i}\left( F_{l_i}\left(
x_{_{li}}\right) ,F_{m_i}\left( x_{_{m_i}}\right) \right)$ are assigned to
the edges of the tree $T_{1}$.

The
copula densities
$
c_{a_{ij}^2,b_{ij}^2|S_{ij}^2}\left( F_{a_{ij}^2|S_{ij}^2}\left(
x_{a_{ij}^2}|\mathbf{x}_{S_{ij}^2}\right) ,F_{b_{ij}^2|S_{ij}^2}\left(
x_{b_{ij}^2}|\mathbf{x}_{S_{ij}^2}\right) \left | \;\mathbf{x}_{S_{ij}^2} \right. \right) 
$
are assigned to each pair clusters $e_i^2$ and $e_j^2$ , which are linked in
the junction tree $T_2$, where:
\[
S_{ij}^2 = e_i^2\cap e_j^2,\hspace{2mm}
a_{ij}^2 = e_i^2-S_{ij}^2, \hspace{2mm}
b_{ij}^2 = e_i^2-S_{ij}^2.
\]

The copula densities
$
c_{a_{ij}^k,b_{ij}^k|S_{ij}^k}\left( F_{a_{ij}^k|S_{ij}^k}\left(
x_{a_{ij}^k}|\mathbf{x}_{S_{ij}^k}\right) ,F_{b_{ij}^k|S_{ij}^k}\left(
x_{b_{ij}^k}|\mathbf{x}_{S_{ij}^k}\right) \left | \;\mathbf{x}_{S_{ij}^k} \right. \right) 
$
are assigned to each pair of clusters $e_i^k$ and $e_j^k$, which are linked
in the $T_k$ junction tree, where:
$S^k = e_i^k\cap e_j^k, \:
a_{ij}^k = e_i^k-S_{ij}^k, \:
b_{ij}^k = e_i^k-S_{ij}^k.$
It is easy to see that $a_{ij}^k$ and $b_{ij}^k$ contain a single element only.

\begin{thm}\label{theo:theo4.2}
The Regular-vine probability distribution associated to the cherry-vine structure given in Definition \ref{def:def4.1} can be expressed as:
\begin{equation}\label{eq:eq9}
\begin{array}{l}
f\left( x_1,\ldots ,x_d\right) =\left[ \prod\limits_{i=1}^df_i\left(
x_i\right) \right] \left[ \prod\limits_{i=1}^{d-1}c_{e_i^1}\left( F_{l_i}\left(
x_{_{li}}\right) ,F_{l_i}\left( x_{_{li}}\right) \right) \right]\\
\cdot \prod\limits_{i=2}^{d-1}\prod\limits_{e\in
E_i}c_{a_{ij}^k,b_{ij}^k|S_{ij}^k}\left( F_{a_{ij}^k|S_{ij}^k}\left(
x_{a_{ij}^k}|\mathbf{x}_{S_{ij}^k}\right) ,F_{b_{ij}^k|S_{ij}^k}\left(
x_{b_{ij}^k}|\mathbf{x}_{S_{ij}^k}\right) \left | \;\mathbf{x}_{S_{ij}^k} \right. \right).
\end{array}
\end{equation}
where $F_{a_{ij}^k|S_{ij}^k}$ are defined by Joe (1997) as
\[
F_{j |D}\left( x_{j}|\mathbf{x}_{D}\right) =\frac{\partial C_{i,j|D \backslash
\left\{ i\right\} }\left( u_i,u_j\right) }{\partial u_i}\left |_
{\hspace{-2mm}u_i=F_{i|D \backslash \left\{ i\right\} }\left( x_i|\mathbf{x}%
_{D \backslash \left\{ i\right\} }\right)  \atop %
u_j=F_{j)|D \backslash \left\{ i\right\} }\left( x_{j}|\mathbf{x}_{D \backslash \left\{ i\right\}
}\right)}\right .
\]
for $i\in D, D \subset V$.
\end{thm}

\begin{defn}\label{def:def4.3}
The copula associated to the joint density function $f(x_1,...x_d)$, given in formula (\ref{eq:eq9}) is called
cherry-vine copula.
\end{defn}

We mention here that the so called cherry-vine copula is a Regular-vine copula which does not use any of the conditional independences.

\section{Truncated R-vine as a special case of cherry-tree copula}
\label{sec:sec5}
\smallskip

As the number of variables grows, the number of conditional pair copulas grows rapidly. For example in (Dissman et al. (2013) for 16 variables the number of pair copulas involved, which have to be modeled and fitted is $120= 15+14+\cdots +2+1$. To keep such structure tractable for inference and model selection, the simplifying assumption that copulas of conditional distributions do not depend on the variables which they are conditioned on is popular. Although this assumption leads in many cases to misspecifications as it is pointed out in Acar et al. (2012) and in Hobaek Haff and Segers (2010). In Hobaek Haff et al. (2010) are presented classes of distributions where simplification is applicable.
An idea to overcome the fitting of a large number of pair copulas with large conditioning set is to exploit the conditional independences between the random variables. This idea was already discussed for Gaussian copulas in Kurovicka and Cooke (2006), based on the idea inspired by Whittaker (1990). Our approach presented here is more general.

In the following remark Aas et al (2009) give the relation between conditional independences and conditional pair-copulas.

\begin{rem}\label{rem:rem5.1}
$X_i$ and $X_j$ are conditional independent given the set of variables $%
\mathbf{X}_A, A \subset V\backslash \left\{ i,j\right\} $ if and only if
\[
c_{ij|A}\left( F_{i|A}\left( x_i|\mathbf{x}_A\right) ,F_{j|A}\left( x_j|%
\mathbf{x}_A\right) \left | \;\mathbf{x}_A \right. \right) =1. 
\]
\end{rem}

The following definition of {\it truncated vine at level k} is given in
Brechmann et al. (2012).

\begin{defn}\label{def:def5.3}
A \textit{pairwisely truncated R-vine at level k} (or truncated R-vine at level $k$) is a special R-vine copula with the property that all pair-copulas
with conditioning set equal to, or larger than $k$, are set to bivariate independence copulas.
\end{defn}

In their approach Brechmann et al. (2012), construct the truncated vines by choosing in the first $k$-trees the strongest Kendall-tau between the variables. In the last trees the pair copulas were set to one. We claim that the strong dependences in the lower trees do not imply independences in the last trees in general. This is easy to understand because of the great number of possibilities to build the last trees, starting from the same first trees. 

Another approach, which is much closer to ours, is given in Kurowicka (2011). Her idea was building trees with lowest dependence (conditional independences) in the top trees, starting with the last tree (node). Her method uses partial correlations which in case of Gaussian copula are theoretical well grounded.

Now we prove that a general $k$-width junction tree copula (see Definition \ref{def:def3.2}) can be expressed as a $k$-th order cherry-tree copula.

Let us suppose that we have a $d$-dimensional random vector $\left(
X_1,\ldots ,X_d\right) ^T$ with $k$-width junction tree structure (which
means there are some conditional independences between the random variables
contained by the random vector.) This means that $P\left( \mathbf{X}\right) $
can be written as a junction tree pd of k-with. By theorem 3.1 we can
express the copula density of $P\left( \mathbf{X}\right) $ as a $k$-with
junction tree copula. As the copula density function is a special case of
probability density function our theorems reminded in the preliminaries can
be applied to them.

Since the copula density $c\left( \mathbf{U}_V\right) $ of $P\left( \mathbf{X%
}\right) $ can be expressed as a $k$-with junction tree copula we have 
$
KL\left( c_{k-JT}\left( \mathbf{U}_V\right) ;c\left( \mathbf{U}_V\right)
\right) =0.
$

According to Theorem \ref{theo:theo2.4} we can find a $k$--th order cherry-tree copula
density such that
$
KL\left( c_{k-ChT}\left( \mathbf{U}_V\right) ;c\left( \mathbf{U}_V\right)
\right) \leq KL\left( c_{k-JT}\left( \mathbf{U}_V\right) ;c\left( \mathbf{U}%
_V\right) \right) =0, 
$
where $c_{k-ChT}\left( \mathbf{U}_V\right) $ denotes the $k$-th order cherry
tree copula density.

Since the Kulback-Leibler divergence (Cover and Thomas (1991)) is
greater than or equal $0$ it follows that 
$
KL\left( c_{k-ChT}\left( \mathbf{U}_V\right) ;c\left( \mathbf{U}_V\right)
\right) =0\text{.} 
$

From this it follows that $c\left( \mathbf{U}_V\right) $ can be expressed as a $k$-th
order cherry-tree copula. This is why the $k$-th order cherry-tree copulas are useful. 

\begin{thm}\label{theo:theo5.5}
Every $k$-th order cherry tree copula, associated to a cherry tree pd,
can be expressed as a $(k+1)$-th order cherry tree copula.

\end{thm}

\begin{proof}
According to Theorem \ref{theo:theo2.6} in Section \ref{sec:sec2} starting from
the $k$-th order cherry- tree copula we can obtain a $\left( k+1\right) $-th
order cherry- tree copula, which gives at least as good approximation to $%
c\left( \mathbf{U}_V\right) $, as the $k$-th order did.

This is that
$
KL\left( c_{\left( k+1\right) -ChT}\left( \mathbf{U}_V\right) ;c\left( 
\mathbf{U}_V\right) \right) \leq KL\left( c_{k-JT}\left( \mathbf{U}_V\right)
;c\left( \mathbf{U}_V\right) \right) =0, 
$
and taking again into account that KL divergence is greater than or equal $0$ it
follows that 
$
KL\left( c_{\left( k+1\right) -ChT}\left( \mathbf{U}_V\right) ;c\left( 
\mathbf{U}_V\right) \right) =0.
$
\end{proof}

This theorem implies the following result.

\begin{thm}\label{theo:theo5.7}
The truncated vine at level $k$ is a $k$-th order
cherry-tree copula.
\end{thm}

This theorem suggests that, for searching for the best fitting truncated vine at level is useful to search for the best fitting cherry tree copula. 

An important property of the cherry tree copula structures is that these are capable to reduce massively the number of conditional copulas, and also the dimension of the conditioning set.

\smallskip
{\footnotesize
\hspace*{0.5cm}

\begin{minipage}[t]{8cm}$$\begin{array}{l}
\mbox{Edith Kov\'{a}cs -- Department of Mathematics and Informatics,}\\
\mbox{Budapest College of Management},\\
 \mbox{Vill\'{a}nyi \'{u}t 11-13., Budapest, 1114 HUNGARY}\\
\mbox{E-mail: kovacs.edith@avf.hu}\end{array}$$
\end{minipage}}

\begin{thebibliography}{99}

\bibitem{} K. Aas, C. Czado, A. Frigessi, and H. Bakken, Pair-copula constructions of multiple dependence, Insur. Math. Econ., 44, 182--198, (2009)

\bibitem{} E.F. Acar, C. Genest and J. Neslehova, Beyond simplified pair-copula constructions, Journal of Multivariate Analysis, 110, 74--90, (2012)

\bibitem{} T. Bedford and R. Cooke, Probability density decomposition for conditionally dependent random variables modeled by vines, Ann. Math. Artif. Intell., 32, 245--268, (2001)

\bibitem{} T. Bedford and R. Cooke, Vines -- a new graphical model for dependent random variables, Ann. Stat., 30(4), 1021--1068, (2002)

\bibitem{} E.C. Brechmann, C. Czado and K. Aas, Truncated regular vines in high dimensions
with applications to financial data, Canadian Journal of Statistics, 40(1), 68--85, (2012)

\bibitem{} J. Buksz\'{a}r and A. Pr\'{e}kopa, Probability Bounds with Cherry Trees, Mathematics of Operational Research, 26, 174--192, (2001)

\bibitem{} J. Buksz\'{a}r and T. Sz\'{a}ntai, Probability Bounds Given by Hypercherry Trees, Optimization Methods and Software, 17, 409--422, (2002)

\bibitem{} C. Czado, Pair-copula constructions of multivariate copulas, In: P. Jaworski, F. Durante, W. H\"{a}rdle and T. Rychlik (Eds.), Copula Theory and Its Applications, Berlin, Springer, (2010)

\bibitem{} J. Dissman, E.C. Brechmann, C. Czado and D. Kurowicka, Selecting and estimating regular vine copulae and application to financial returns,  Computational Statistics and Data Analysis, 59, 52--69, (2013)

\bibitem{} I. Hobaek Haff, K. Aas and A. Frigessi, On the simplified pair-copula construction -- simply useful or too simplistic? Journal of Multivariate Analysis, 101(5), 1296--1310, (2010)

\bibitem{} I. Hobaek Haff and J. Segers, Nonparametric estimation of pair-copula constructions with the empirical pair-copula, http://arxiv.org/abs/1201.5133, (2010)

\bibitem{} H. Joe, Multivariate Models and Dependence Concepts, Chapman \& Hall, London, (1997)

\bibitem{} D. Kurowicka and R. M. Cooke, Uncertainty Analysis with High Dimensional Dependence Modelling, Chichester, John Wiley, (2006)

\bibitem{} D. Kurowicka, Optimal truncation of vines, in: D. Kurowicka and H. Joe (eds) Dependence-Modeling -- Handbook on Vine Copulas, Word Scientific Publishing, Singapore, (2011)

\bibitem{} S.L. Lauritzen and D.J. Spiegelhalter, Local Computations with Probabilites
on Graphical Structures and their Application to Expert Systems, J.R.
Statist. Soc. B, 50, 157--227, (1988)

\bibitem{} Sz\'{a}ntai, T. and E. Kov\'{a}cs, 
Hypergraphs as a mean of discovering the dependence structure of a discrete multivariate probability
distribution, \textit{Proc. Conference APplied mathematical programming and 
MODelling (APMOD), 2008}, Bratislava, 27-31 May 2008,
Annals of Operations Research, 193(1), 71--90, (2012)

\bibitem{} J. Whittaker, Graphical Models in Applied Multivariate Statistics, John Wiley \& Sons, (1990)

\end{thebibliography}
\end{document}